\documentclass[fullpage]{article}
\usepackage[utf8]{inputenc}
\usepackage{amsmath}
\usepackage{amsthm}
\usepackage{amssymb}
\usepackage{bbm}
\usepackage{mathtools}
\usepackage{amsfonts}
\usepackage{amstext}
\usepackage[shortlabels]{enumitem}
\usepackage{subfig}
\usepackage{float}
\usepackage{ytableau}
\usepackage{physics}
\usepackage{esint}
\usepackage{fancyhdr}
\usepackage{tensor}
\usepackage{xcolor}
\usepackage{mathtools}
\usepackage{subfig}
\usepackage{comment}
\usepackage{cancel}
\usepackage{tensor}
\usepackage{mathrsfs}
\usepackage{cancel}
\usepackage{leftindex}
\usepackage{mathtools}
\usepackage{graphicx}
\usepackage{blindtext}
\usepackage{quiver}
\graphicspath{ {./Desktop/} }
\usepackage[a4paper,bindingoffset=0.2in,%
            left=0.75in,right=0.75in,top=1in,bottom=1in,%
            footskip=.25in]{geometry}
\usepackage{blindtext}
\everymath{\displaystyle}
\usepackage[shortlabels]{enumitem}
\usepackage{hyperref}
\usepackage{scalerel}
\usepackage{setspace}

\newgeometry{left=1in,right=1in,top=1.5in,bottom=1.5in}
\newcommand{\fixme}[1]{{{\color{blue}{[#1]}}}}

\newcommand{\mf}[1]{\mathfrak{#1}}
\newcommand{\Sym}[1]{\text{Sym}#1}

\newcommand{\be}{\begin{equation}}
\newcommand{\ee}{\end{equation}}
\newcommand{\btik}{\begin{tikzcd}}
\newcommand{\etik}{\end{tikzcd}}
\newcommand{\End}{\mathrm{End}}
\newcommand{\Hom}{\mathrm{Hom}}
\newcommand{\Mod}{\mathrm{-Mod}}
\newcommand{\tgv}{\mathfrak{d}_{\mathfrak g, V}}

\newcommand{\HR}{\CM_{K, T^*V}}

\newcommand{\hypquotient}{/\!\!/\!\!/}

\newenvironment{amatrix}[1]{\left[\begin{array}{@{}*{\numexpr#1-1}{c}|c@{}}}{\end{array}\right]}

\makeatletter
\newcommand{\tmod}[1]{{\@displayfalse\mod{#1}}}
\makeatother

\makeatletter
\renewcommand*\env@matrix[1][*\c@MaxMatrixCols c]{%
    \hskip -\arraycolsep
    \let\@ifnextchar\new@ifnextchar
    \array{#1}}
\makeatother

\newcommand{\C}{\mathbb C}

\newcommand{\Z}{\mathbb Z}

\newcommand{\E}{\mathbb E}

\newcommand{\fd}{\mathfrak{d}}
\newcommand{\fg}{\mathfrak{g}}
\newcommand{\fh}{\mathfrak{h}}

\newcommand{\CC}{{\mathcal C}}

\newcommand{\CE}{{\mathcal E}}

\newcommand{\CL}{{\mathcal L}}
\newcommand{\CM}{{\mathcal M}}
\newcommand{\CN}{{\mathcal N}}

\newcommand{\CT}{{\mathcal T}}

\newtheorem{Def}{Definition}[section]
\newtheorem{Thm}[Def]{Theorem}
\newtheorem{Prop}[Def]{Proposition}
\newtheorem{Cor}[Def]{Corollary}
\newtheorem{Lem}[Def]{Lemma}
\newtheorem{Rem}[Def]{Remark}

\newtheorem{Exp}[Def]{Example}

\numberwithin{equation}{section}

\title{Higgs Branches in the Omega-background via the Category of Line Operators}
\author{Thomas Karabela \& Wenjun Niu}
\date{\today}

\begin{document}

\maketitle

\begin{abstract}

 The vacuum manifold $\CM$ of a topological twist of a 3d $\CN=4$ gauge theory is a hyper-K\"ahler variety; deformations and quantizations of $\CM$ can be constructed in the framework of 3 dimensional topological quantum field theories. In particular, based on physics arguments, turning on an omega-background results in the quantization of $\C[\CM]$ as a Poisson algebra. In this paper, we implement this idea mathematically, in the context of the B-twist of 3d $\CN=4$ gauge theories, namely for Higgs branches. Our strategy is to implement omega-background in the category of line operators via the choice of a ribbon twist, and obtain the quantum Hamiltonian reduction as derived endomorphisms in the equivariant category, the category where the ribbon twist acts trivially. We apply quadratic Koszul duality to perform this computation. 
 
\end{abstract}

\tableofcontents

\section{Introduction}\label{sec:intro}

\subsection{Quantization and omega-background}

Let $G$ be a reductive Lie group and $V$ a finite-dimensional representation. The pair $(G, V)$ determines a 3d $\CN=4$ supersymmetric gauge theory of cotangent type, with gauge group $G_c$ (the compact part of $G$) and matter-fields valued in $T^*V$. This gauge theory admits two topological twists, the so-called $A$ and $B$ twist. The vacuum manifold for the B twist is called the Higgs branch $\CM_H$ and that for the A twist is called the Coulomb branch $\CM_C$. These are hyper-Kh\"ahler varieties and contain important representation-theoretic data. The variety $\CM_H$ is defined as the symplectic reduction $T^*V\hypquotient G$, while $\CM_C$ is defined in the famous work \cite{braverman2018towards} via the moduli space of triples, or the BFN space. Moreover, the famous 3d mirror symmetry \cite{intriligator1996mirror} suggests that there should be an involution acting on 3d $\CN=4$ theories such that when there is a pair of mirror dual gauge theory $(G, V)\leftrightarrow (G^!, V^!)$, there is an isomorphism of hyper-Kh\"ahler varieties $\CM_H\cong \CM_C^!$. 

The advantage of realizing $\CM_i$ for $i= H, C$ as vacuum varieties is that one can apply intuition from physics to understand the geometry of these varieties. For example, the resolutions of $\CM_i$ are related to the mass and FI deformations in the gauge theory. The other example, which is the theme of this work, is that the quantization of $\C[\CM]$ can be achieved by turning on the omega-background. See, for example, \cite[Section 6]{beem2020secondary} or \cite{yagi2014omega}. 

The idea is that $\C[\CM]$, being the algebra of local operators in a 3d TQFT, is a commutative Poisson algebra. Physically, this commutativity is due to the 3-dimensional nature of the multiplication, and the Poisson structure is a homotopical structure typical in twisted gauge theories. In more mathematical terms, local operators form an $\E_3$ algebra, whose cohomology is a $\mathbb{P}_2$ algebra (2-shifted commutative Poisson algebra). Suppose we now choose an axis in 3d and turn on an $S^1$-rotation perpendicular to this direction, then local operators become non-commutative: operators cannot pass each other since they are forced to stay on the axis. The resulting algebra is expected to be a deformation of the Poisson algebra structure on $\C[\CM]$. This action of $S^1$ is the omega-background in 3d TQFT context. 

This is a very elegant physical picture of deformation quantization. However, it is tricky to implement, to say the least. The action of $S^1$ on $\C[\CM]$ is homotopical; therefore, it is difficult to take its equivariants and compute the remaining algebra structure. In this paper, we implement this picture using the category of line operators. It turns out that by studying this category, we significantly simplify the action of $S^1$ and the question of taking $S^1$-equivariants. 

More precisely, let $\CL$ be the category of line operators. Being the category of line defects in a 3d TQFT, $\CL$ has the structure of a braided tensor category. Let $\mathbbm 1\in \CL$ be the unit object; then\footnote{In the setting of the topological twists of 3d $\CN=4$ gauge theories, $\CL$ is the derived category of some abelian category, and the endomorphism is a derived endomorphism. }
\be
\C[\CM]\cong\End_{\CL}(\mathbbm 1).
\ee
Let us turn to the question of $S^1$-action. We can put line operators along the axis fixed by $S^1$, then one obtains an action of $S^1$ on $\CL$. This action is well-known to be given by the structure of a \textit{ribbon twist} $\theta$, namely an autoequivalence of the identity functor $\mathrm{Id}_\CL$ of $\CL$ that satisfies
\be
\theta_{L_1\otimes L_2}=c_{L_2, L_1}c_{L_1, L_2} \theta_{L_1}\otimes \theta_{L_2}, \qquad \theta_{L^\vee}=\theta_L^\vee, \qquad \theta_{\mathbbm 1}=1, 
\ee
for all $L_1, L_2, L\in \CL$. Here $c_{L_1, L_2}$ is the braiding and $L^\vee$ is the rigid dual\footnote{We will assume $\CL$ is rigid. One can drop the second condition if $\CL$ is not rigid.}. Note that we don't need homotopy theory here at all to define this action! A better way to say this is that the homotopy data is hidden in the categorical data of $\CL$. The trade-off is, of course, that one needs to deal with extended objects, but that is significantly easier in our opinion. 

To be able to put a line operator in the omega-background, one requires that this object is invariant under $S^1$ rotation. This happens precisely when $\theta=1$. We find that line operators in the omega-background is roughly
\be
\CL^{S^1}=\CL/(\theta-1).
\ee
This can be made mathematically precise using derived geometry techniques \cite{preygel2011thom}, as long as one does everything in a sufficiently DG/derived manner. By definition, $\mathbbm 1\in \CL^{S^1}$. The physics intuition of omega-backgrounds then suggest that
\be
\End_{\CL^{S^1}}(\mathbbm 1)
\ee
is a deformation quantization of $\C[\CM]$. 

Having translated the physics of omega-backgrounds into a mathematically accessible form, let us explain the setting to which we apply this strategy: B-twist of 3d $\CN=4$ gauge theories, in which case $\CM=\CM_H=T^*V\hypquotient G$, as we have discussed. In this case, it is already known that $\C[\CM_H]$ admits a quantization via the quantum Hamiltonian reduction construction; in fact, this is the quantization we get from the omega-background procedure. This could serve as a sanity check of the physics proposal. Moreover, we hope that this idea can be used to study quantization of Poisson vertex algebras by studying extended defects in 4 dimensional holomorphic-topological theories.

\subsection{Line defects in the B-twist and Lie superalgebras}

We now focus on the topological B-twist of a 3d $\CN=4$ gauge theory of cotangent type, defined by $(G, V)$. Categorical aspects of line operators in the B twists of 3d $\mathcal N=4$ gauge theories were developed in a series of different works, including \cite{oblomkov2024categorical, costello2019vertex, costello2019higgs, hilburn2023tate, BCDN, gammage2025betti, creutzig2025kazhdan, dimofte2024tannakian, niu2024quantum}. For an  interested reader, see the introductions to \cite{dimofte2024tannakian} for further background and references. We will follow the approach of \cite{niu2024quantum} in studying $\CL$.

Let $\fg\ltimes V[-1]$ be the semi-direct product Lie algebra, where $V$ is in degree $1$, and let $\tgv:=T^*(\fg\ltimes V[-1])$, its cotangent Lie algebra. This Lie algebra appeared in \cite{costello2019vertex} in the construction of boundary vertex algebras for the B-twist. In \cite{niu2024quantum}, the second author clarified the relation between $\tgv$ and $\CM_H$: $\tgv$ is the \textit{tangent Lie algebra} of $\CM_H$, and, moreover, $\C[\CM_H]$ is identified with the (relative) Chevalley-Eilenberg cochian complex of $\tgv$. 

This is to be expected. It is well-known to experts that the B-twist of the gauge theory defined by $(G, V)$ is analogous to the Chern-Simons theory for $\tgv$. Following Drinfeld \cite{drinfeld1989quasi}, we take the category $\CL$ to be (the derived category of) $U(\tgv)\Mod^G$, the category of modules of the Lie algebra $\tgv$ where the action of $\fg$ integrates to an action of the group $G$. The braided tensor structure is supplied by solutions of KZ equations associated with the canonical perfect pairing of $\tgv$. This is consistent with the conformal field theory approach of \cite{costello2019vertex, costello2019higgs}. This is \textbf{not} the full category of line operators, but is rather the subcategory of purterbative line operators. However, this is enough for the problem of computing $\C[\CM_H]$, since this is a perturbative question. In conclusion, $\CL=U(\tgv)\Mod^G$.

To specify the $S^1$ action, we need to choose a ribbon element. Let $C$ be the quadratic Casimir element associated with the canonical pairing of $\tgv$, then $\theta=e^C$ provides a ribbon element and therefore an action of $S^1$. In general, one can choose a central character $J\in \fg^*[-2]$ and consider $\theta=e^{C+J}$. This does not satisfy $S\theta=\theta$, but still defines an action of $S^1$ on $\CL$. The category $\CL^{S^1}$ of line operators in the omega-background is therefore (the derived category of)
\be
U(\tgv)/(\theta-1)\Mod^G\simeq U(\tgv)/(C+J)\Mod^G,
\ee
where the second equivalence is due to the fact that $C+J$ is nilpotent; $\theta=1$ iff $C+J=0$. We can now translate the quantization by omega-background idea into the following very concrete statement
\be
\End_{U(\tgv)/(C+J)\Mod^G} (\C) \text{ is a deformation quantization of } \C[\CM_H]. 
\ee
This is the statement that we shall check in the rest of the paper, up to some mathematical nuances. In fact, we will show that this deformation quantization is the usual quantum Hamiltonian reduction. 

\subsection{Mathematical statements and strategy}\label{subsec:mathintro}

For simplicity, we will remove the requirement of the integrability of the $G$-action. One can easily recover this with a relative Chevalley-Eilenberg cochain complex. We will also replace the quotient with its DG model $U(\tgv)[\psi]$, with a differential $d\psi=C+J$ and $\psi$ of degree $1$. Let $\widehat\CM_H$ be given by $T^*V\hypquotient\widehat G$, where $\widehat G$ is the formal group. The only difference is that one takes the Lie algebra invariants of $\mu^{-1}(0)$. We consider an action of $\C^\times$ where the weights of $U(\tgv)[\psi]$ is equal to its cohomological grading. Our main result is the following statement.

\begin{Thm}\label{Thm:main}

     Let $\C$ be the trivial module of $U(\tgv)[\psi]$, then
    \be
\CC\CE (\tgv, J):= \mathrm{REnd}_{U(\tgv)[\psi]\Mod}^{\C^\times}\left( \C \right)
    \ee
     is a deformation quantization of the Poisson algebra $\C[\widehat \CM_H]$. It is identified with the usual quantum Hamiltonian reduction with respect to the same central character $J$. Here, $\mathrm{REnd}_{U(\tgv)[\psi]\Mod}^{\C^\times}$ denotes the space of $\C^\times$-enriched endomorphisms. 
    
\end{Thm}

Let us explain in more detail the contents of the above theorem. Since the trivial module $\C$ factors through the map $U(\tgv)[\psi]\to \C[\psi]$, there is a map of algebras
\be
\C[\varphi]=\mathrm{REnd}_{\C[\psi]\Mod}^{\C^\times}\left(\C\right)\to \CC\CE(\tgv, J).
\ee
Our statement is that in an explicit DG model of the algebra $\CC\CE (\tgv)$, 

\begin{enumerate}

    \item the map above is central;

    \item the algebra $\CC\CE (\tgv)\otimes_{\C[\varphi]}\C$ is commutative;
    
    \item it is moreover identified with $\C[\widehat\CM_H]$ as a Poisson algebra. 
    
\end{enumerate}

We will compute the explicit DG model of $\CC\CE (\tgv, J)$ by finding an explicit resolution of $\C$ as a module of $U(\tgv)[\psi]$. This resolution is achieved through quadratic duality. We will review the quadratic and Koszul duality more closely in Section \ref{sec:koszul}. Here, we simply outline the strategy with which we calculate $\CC\CE (\tgv, J)$. It turns out that the output of the Koszul duality calculation is simply the quantum Hamiltonian reduction. 

If $A$ is a quadratic algebra generated by $V$ and the relation $R\subseteq V^{\otimes 2}$, then its quadratic dual $A^!$ is generated by $V^!:=V^*[-1]$ with relation $R^\perp\subseteq (V^*)^{\otimes 2}$, where $(R, R^\perp)=0$. The vector space $A\otimes (A^!)^*$ admits a square-zero differential $\sum (x_n)_R\otimes (y^n)_L$, where $\{x_n\}\subseteq V$ is a basis with dual basis $\{y^n\}$. According to \cite{BGSkoszul}, the algebra $A$ is Koszul if this complex is a resolution of $\C$ as a left $A$ module. Note, however, that one needs to be careful with the definition of $(A^!)^*$ when it is not finite-dimensional, and the solution is usually to introduce an extra $\C^\times$ action such that weight spaces are finite-dimensional, and we take the graded dual. In the case where this happens, we find that $A^!=\mathrm{REnd}_{A\Mod}^{\C^\times} (\C)$, where $\mathrm{REnd}_{A\Mod}^{\C^\times}$ means $\C^\times$-enriched Hom. 

In our case, the algebra $U(\tgv)[\psi]$ is not quadratic. We consider the quadratic algebra $A=U_\hbar (\tgv)[d, \psi]$ with relation $\{d, \psi\}=C+\hbar J$. Here $U_\hbar (\tgv)$ is the associated Reez algebra of $U(\tgv)$ over $\C[\hbar]$, which is quadratic. We can then easily compute the quadratic dual $A^!$ and the corresponding algebra map $\C[\psi]^!=\C[\varphi]\to A^!$. We will show:
\begin{enumerate}

    \item the generator $\hbar^!$ acts as a square-zero differential in $A^!$ and $d^!, \varphi=\psi^!$ are central elements;

    \item the DG algebra generated by $\tgv^!\oplus d^!\oplus \varphi$ with differential $\{\hbar^!,-\}$ and relation $d^!=1$ is identified with the quantum Hamiltonian reduction, which is an explicit quantization of $\C[\widehat \CM_H]$; this algebra is our $\CC\CE(\tgv, J)$;

    \item the product $U(\tgv)[\psi]\otimes \CC\CE(\tgv, J)$ can be deformed by an MC element to become a free resolution of $\C$. 
    
\end{enumerate}

We will construct this resolution by performing periodic localization on the Koszul complex of $A\otimes A^!$.

\begin{Rem}

    In the above and what follows, we always work in the differential-graded setting over the field $\C$. In Particular, all vector spaces are graded, and Koszul sign rule is assumed when swapping factors in tensor products of graded vector spaces. When an extra $\C^\times$ action is used, we refer to the grading by this $\C^\times$ action as weights. 
    
\end{Rem}

\subsection*{Acknowledgements}

We'd like to thank Kevin Costello for his guidance and support, and for Perimeter Institute for providing a wonderful working environment. T.K. wishes to acknowledge the financial support from the  Natural Sciences and Engineering Research Council of Canada in the form of the Undergraduate Student Research Award grant  USRA-604504-2025. W.N. would like to thank Lukas M\"uller for discussions, as well as his friend Don Manuel for many encouragements. The research of W.N. is supported by Perimeter Institute for Theoretical Physics. Research at Perimeter Institute is supported in part by the Government of Canada through the Department of Innovation, Science and Economic Development Canada and by the Province of Ontario through the Ministry of Colleges and Universities.

\section{Symplectic reductions and braided tensor categories}\label{sec:higgs}

\subsection{Symplectic reductions and quantizations}\label{subsec:higgsquan}

We fix a reductive algebraic group $G$ and a finite-dimensional representation $V$. Let $\fg=\mathrm{Lie}(G)$. The action of $G$ on $V$ induces an action of $G$ on $T^*V$. The action of $G$ on $T^*V$ preserves the canonical symplectic form, and there is a moment map:
\be
\mu: T^*V\to \fg^*
\ee
such that:
\be
\mu(v,v^*)(X)=(Xv,v^*),~\text{ for all } X, v, v^*.
\ee
The map $\mu$ is a $G$-equivariant and generically flat, but in general not smooth. The zero fibre $\mu^{-1}(0)$ is a $G$-invariant subspace. Note that when $\mu$ is not flat, we take $\mu^{-1}(0)$ to be the \textbf{derived fiber}, namely the derived (affine) scheme represented by the following Cartesian diagram:
\be
\btik
\mu^{-1}(0)\rar\dar & 0\dar\\
T^*V\rar{\mu} & \fg^*
\etik
\ee 

\begin{Def}\label{Def:HamRed}
The symplectic reduction $\CM_H$ is defined as the affine scheme:
\be
\CM_H:=\mu^{-1}(0)/\!/G,
\ee
namely $\CM_H$ is an affine variety whose space of algebraic functions is the algebra $\C[\mu^{-1}(0)]^G$. We also denote the above by $T^*V\hypquotient G$, called the hyper-K\"{a}hler quotient. 
\end{Def}

The affine scheme $\mu^{-1}(0)$ is represented by the following explicit commutative differential graded algebra. The underlying algebra is the algebra of functions on $T^*V\oplus \fg^*[-1]$. The differential is defined on generators of $\C[\fg^*[-1]]$ by $d\mu$. Explicitly, let $\{b_i\}$ be a set of linear coordinate functions on $\fg^*[-1]$ (namely, a basis for $\fg$, in degree $-1$), then the differential is given by:
\be
d b_i=(x_iv, v^*). 
\ee
This DG algebra has a natural action of $G$ and $\C[\CM_H]$ is the $G$-invariant sub-algebra. 

If we calculate the $\fg$-invariants instead of the $G$-invariants, we introduce another set of generators $\{c^i\}$, which form a basis in $\fg^*[-1]$, and introduce a differential
\be\label{eq:differential}
\begin{aligned}
    db_i&= -\sum f_{ij}^k c^jb_k-(x_iv, v^*),\qquad d \beta_m= -\sum c^i(\rho (b_i)\rhd \beta_m), \\ &d\gamma^n=-\sum c^i(\rho (b_i)\rhd \gamma^n) ,\qquad dc^i=-\frac{1}{2}\sum f^i_{jk}c^j c^k. 
\end{aligned}
\ee
Here $\beta_m, \gamma^m$ are linear coordinates of $T^*V$. Let us denote the resulting commutative DG algebra by $\C[\widehat \CM_H]$. This is a Poisson algebra, with Poisson bracket
\be\label{eq:Poisson}
\{\beta_m, \gamma^n\}=-\delta_{m}^n, \qquad \{b_i, c^j\}=-\delta_i^j. 
\ee
It is clear from the definition that $\C[\CM_H]$ is a Poisson subalgebra of $\C[\widehat\CM_H]$ via embedding 
\be
\C[\mu^{-1}(0)]^G\to \C[\mu^{-1}(0)]\to \C[\widehat \CM_H].
\ee
Note that the second embedding is not one of DG algebras. The composition is a map of DG algebras, because if a function $f$ is $H$-invariant, then $d_{CE}f=0$. 

The Poisson algebra $\C[\CM_H]$ and $\C[\widehat \CM_H]$ admits quantizations via the \textit{quantum Hamiltonian reduction} construction. More explicitly, in the case of $\C[\CM_H]$, one considers the Weyl algebra $\C_\varphi [\beta_m, \gamma^n]$ with commutation relation $[\beta_m, \gamma^n]=-\delta^n_m\varphi$. The moment map can be quantized into an algebra map $\mu: U_\varphi(\fg)\to \C_\varphi[\beta_m, \gamma^n]$. Let $J\in \fg^*$ be a character, and define $I_{\mu, J}:=A\cdot \{\mu(x)+\varphi J(x), x\in \fg\}$. The vector space
\be
\C_\varphi^J[\CM_H]:= \left(\C_\varphi[\beta_m, \gamma^n]/I_{\mu, J}\right)^G
\ee
inherits an algebra structure and is called the quantum Hamiltonian reduction. This is an explicit quantization of $\C[\CM_H]$. One can replace Lie groups by Lie algebras and obtain $\C_\varphi^J[\widehat\CM_H]$ as a quantization of $\C[\widehat \CM_H]$. Setting $\varphi=0$ recovers the original Poisson algebra.

The algebra $\C_\varphi^J[\widehat\CM_H]$ admits an explicit DG model, whose underlying algebra is simply $\C[\beta_m, \gamma^n]\otimes \C[b_i, c^j]$ with commutation relation
\be
[\beta_m, \gamma^n]=-\varphi\delta_m^n, \qquad [b_i, c^j]=-\varphi\delta_i^j,
\ee
and a differential that has the same form as equation \eqref{eq:differential} except that
\be
db_i= -\frac{1}{2}\sum f_{ij}^k (c^j b_k-b_k c^j)-(x_iv, v^*)- J_i \varphi. 
\ee
This is the DG algebra that will be matched with $\CC\CE(\tgv, J)$.

\subsection{Symplectic reductions and Lie algebras}\label{subsec:tangentLie}
\label{symred}

From $(G, V)$, we obtain a graded Lie algebra $\fh:=\fg\ltimes V[-1]$. If we fix basis $\{x_i\}\subseteq \fg$ and $\{\theta_m\}\subseteq V[-1]$, then the nontrivial commutation relations are
\be\label{eq:commuth}
[x_i, x_j]=\sum f_{ij}^k x_k, \qquad [x_i, \theta_m]=\sum_n \rho (x_i)_m^n \theta_n.
\ee
Here $\rho: \fg\to \End (V)$ is the action of $\fg$ on $V$. Let $\tgv:= T^*[-2]\fh=\fh\ltimes \fh^*[-2]$. This is a graded Lie algebra with a symmetric bilinear form $\kappa$ of degree $2$. This makes $\tgv$ a perfect Lie algebra. If we denote by $\{\overline\theta^n\}\subseteq V^*[-1]$ and $\{t^i\}\subseteq \fg^*[-2]$ sets of dual basis, then the remaining nontrivial commutation relations in $\tgv$ are
\be\label{eq:commutd}
[x_i, t^j]=\sum -f_{ik}^j t^k,\qquad [x_i, \overline \theta^n]=-\rho(x_i)^n{}_m \overline \theta^{m},\qquad [\theta_m, \overline \theta^n]=\sum_i \rho^n{}_m(x_i)t^i.
\ee
The relation between symplectic reduction $\widehat \CM_H$ was explored in \cite{niu2024quantum}, based on works \cite{roberts2010rozansky, costello2019vertex}. Most importantly, we have the following result. 

\begin{Prop}

    There is an identification of DG algebras
    \be
\mathrm{CE}^*(\tgv)=\C[\widehat\CM_H],
    \ee
    such that the Poisson structure of $\C[\widehat\CM_H]$ is symplectic and is induced by the perfect pairing $\kappa$. Here, $\mathrm{CE}^*(\fg)$ denotes the Chevalley-Eilenberg cochain complex of a DG Lie algebra $\fg$. 
    
\end{Prop}

\begin{proof}

    The identification of algebras is clear (\textit{c.f.} \cite[Lemma 4.5]{niu2024quantum}). This identification implies that the tangent complex $\CT_{\widehat\CM_H}$ of $\widehat \CM_H$ is precisely $\tgv[1]\otimes \C[\widehat \CM_H]$ \cite{hennion2018tangent}. The bilinear form $\kappa$ induces a symplectic structure, and therefore an identification $\CT_{\widehat\CM_H}\cong \CT_{\widehat\CM_H}^*$. The corresponding Poisson structure is easily identified with the one in equation \eqref{eq:Poisson}. 
    
\end{proof}

We would like to explore the relation between quantizations $\C_\varphi^J[\widehat\CM_H]$ of $\C[\widehat\CM_H]$ and the Lie algebra $\tgv$. As we have explained in the introduction, our approach utilizes the ribbon tensor category constructed from $\tgv$ and the associated $S^1$-equivariant category. 

\subsection{Quantized universal enveloping algebras}

The construction of the braided tensor category associated to $\tgv$ is via Drinfeld's quantization of perfect Lie algebras into quasi-triangular quasi-Hopf algebras.  

\begin{Def}

A quasi-triangular quasi-Hopf algebra is a tuple $(A, \Delta, \psi, \Phi, R)$, where:

\begin{enumerate}
\item $A$ is a topological algebra over $\C$, $\Delta:A\to A\otimes A$ and $\psi: A\to \C$ are continuous algebra homomorphisms such that $\psi (1)=1$ and $\Delta (1)=1$. 

\item Invertible elements $R\in   A\otimes A$ and $\Phi\in  A\otimes A \otimes A$. 

\end{enumerate}

These satisfy:
\be
\begin{aligned}
(1\otimes \Delta) \Delta&=\Phi \cdot (\Delta\otimes 1)(\Delta)\cdot \Phi^{-1}\\
(1\otimes 1\otimes \Delta)(\Phi)\cdot (\Delta\otimes 1\otimes 1)(\Phi)&=(1\otimes \Phi)\cdot (1\otimes \Delta\otimes 1)(\Phi\cdot (\Phi\otimes 1)\\
(\psi\otimes 1)\circ \Delta &=1=(1\otimes \psi)\circ \Delta\\
(1\otimes \psi\otimes 1)(\Phi)& = 1\\
(\Delta)^{op}&=R\cdot \Delta\cdot R^{-1}\\
(\Delta\otimes 1)(R)&=\Phi^{312} R^{13}(\Phi^{132})^{-1}R^{23}\Phi,\\
(1\otimes \Delta)(R)&=(\Phi^{231})^{-1}R^{13}\Phi^{213}R^{12}\Phi^{-1}
\end{aligned}
\ee

\end{Def}

\begin{Rem}

    Technically speaking, the definition of a quasi-Hopf algbera includes an antipode $S$. We omit it here since the conditions involving $S$ are technical. In the examples coming from Drinfeld's quantization of perfect Lie algebras, the antipode is the classical antipode. 
    
\end{Rem}

From a quasi-triangular quasi-Hopf algebra, one obtains a rigid braided tensor category $A\Mod$, the category of smooth $A$-modules. 

Let us now fix $\fd=\oplus_n \fd_n$ a finite-dimensional graded Lie algebra with a perfect pairing $\kappa$ of degree $2$. This pairing defermines a Poisson structure on the commutative algebra $\mathrm{CE}(\fd)$, the Chevalley-Eilenbert cochain complex of $\fd$. In applications, $\fd=\tgv$.

Let $C\in U(\fd)$ be the quadratic Casimir and $\Omega=\Delta (C)-C\otimes 1-1\otimes C$. In \cite{drinfeld1989quasi}, Drinfeld constructed a quasi-triangular quasi-Hopf algebra structure on $U(\fd)$, such that $\Delta$ and $\psi$ are the standard symmetric coproduct and counit, but $R=e^\Omega$ and $\Phi_{KZ}$ is the Drinfeld associator, based on solutions of Knizhnik-Zamolodchikov equations. Here we consider a topology where the action of $\fd_n$ for $n\ne 0$ is locally nilpotent. This is not too restrictive, since on any vector space concentrated on finitely many homological degrees, the action of $\fd_n$ is necessarily nilpotent for $n\ne 0$. The condition ensures that $R$ and $\Phi_{KZ}$ are well-defined. 

In particular, the quasi-triangular quasi-Hopf algebra structure induces on the category $U(\fd)\Mod$ the structure of a rigid braided tensor category. This category admits a ribbon structure, given by the element $\theta=e^C$. In \cite{niu2024quantum}, the second author showed that in the case when $\fd=\tgv$, this category admits an alternative description as modules of a quasi-triangular Hopf algebar (a quautum group). The underlying algebra is still $U(\fd)$, but the Hopf structure comes from recognizing this algebra as the double of $U(\fg\ltimes V[-1])$ with its ordinary Hopf algebra structure. We will not use this description here. 

\begin{Rem}
    In fact, for every central element $J\in Z(\fd)$ of degree $2$, one obtains a new ``ribbon element" $\theta_J=e^{C+J}$. It does not satisfy $S\theta=\theta$, but still defines a valid $S^1$ action. 
\end{Rem}

Since $\mathrm{CE}(\fd)=\End_{U(\fg)\Mod}(\C)$, the braided tensor structure of $U(\fd)\Mod$ induces the structure of an $\E_3$ algebra on $\mathrm{CE}(\fd)$. This should be enough to study the $S^1$-equivariance. However, since this structure is very homotopical, it is difficult to extrac a useful information from it. Instead, we consider the calculation in the $S^1$-equivariant category. Let us now choose a twist $\theta=e^{C+J}$, where $C$ is the quadratic Casimir and $J$ a central element of degree $2$. The element $\theta$ induces the action of $S^1=B\Z$ on $U(\fd)\Mod$. By \cite[Lemma 6.15]{preygel2011thom} (the straightforward application of the arguments to $U(\fd)$), we have an equivalence
\be
U(\fd)\Mod^{S^1}\simeq U(\fd)/(\theta-1)\Mod.
\ee
Since in the category $U(\fd)\Mod$, the action of $C+J$ is locally nilpotent, quotienting by $\theta-1$ is equivalent to quotienting by $\log (\theta)=C+J$, namely
\be
U(\fd)\Mod^{S^1}\simeq U(\fd)/(\theta-1)\Mod\simeq U(\fd)/(C+J)\Mod.
\ee
The trivial representation defines an object $\C$ in the $S^1$-equivariant category. We are therefore led to consider the algebra
\be
\CC\CE (\fd, J):=\End_{U(\fd)/(C+J)\Mod}\left( \C\right).
\ee
We would like to calculate this explicitly and verify that $\CC\CE (\fd, J)=\C_\varphi^J[\widehat \CM_H]$, namely that it provides a quantization to the Poisson algebra $\mathrm{CE}^* (\fd)$. We perform this calculation by applying Koszul duality and periodic localization. In the next section, we explain our strategy in greater detail. 

\section{Quadratic Koszul duality}\label{sec:koszul}

In this section, we review quadratic and Koszul duality following \cite{BGSkoszul}. We will need a DG (differential graded) version of this duality, and the results of \textit{op. cit} extends naturally to this setting. (Follow Koszul sign rule.)

\subsection{Quadratic algebras and Koszul algebras}

We first recall the definition of a quadratic algebra.

\begin{Def}
    A quadratic algebra is an algebra $A$ such that there exists a finite-dimensional vector space $V$ and $R\subseteq V\otimes V$ such that
    \be
        A \cong T(V)/(R).
    \ee
    This inherits an $\C^\times$ action from $T(V)$, under which $V$ has weight $1$. 
\end{Def}

\begin{Rem}

    Let us remind readers of our convention that $V$ is a graded (DG) vector space.  We refer to the weights of $\C^\times$ simply as weights, which are not to be confused with grading. 
    
\end{Rem}

We now recall the definition of a Koszul algebra. 

\begin{Def}

Let $A=\oplus_{d\geq 0}A_d$ be a $\C^\times$-equivariant algebra with non-negative weights. We call $A$ a Koszul algebra if $A_0=A/A_{>0}$ is semisimple, and admits a $\C^\times$-equivariant projective resolution
    \be
        \cdots \rightarrow P^2 \rightarrow P^1 \rightarrow P^0 \rightarrow A_0 \rightarrow 0
    \ee
    such that $P^i$ is generated by its weight $i$ component. 
    
\end{Def}

\begin{Rem}

    In this paper, we always assume $A_0=\C$. 
    
\end{Rem}

\begin{Exp}
    Consider the polynomial ring $\C[x] = \textstyle \bigoplus_{j=0}^{\infty}\C x^j$. We have an $\C^\times$-equivariant free $\C[x]$-resolution $\C$
    \be
        0 \rightarrow \C[x]\langle 1\rangle \xrightarrow{x}\C[x]\xrightarrow{0}\C \rightarrow 0
    \ee
  Here $\langle-\rangle$ denotes a shift of both $\C^\times$-equivariant and cohomological degree. The first operation being multiplication by $x$, and the second is the action of $x \mapsto 0$. This shows that it is a Koszul Algebra. 
\end{Exp}

Having established the notions of quadratic and Koszul algebras, and seen that Koszul algebras come equipped with a very structured graded resolution, it is natural to ask what does the dual of such objects look like. This leads us to the definition of the Quadratic dual. Let $A = T(V)/\langle R \rangle$ be a quadratic algebra, where $\langle R \rangle$ is generated by the quadratic relation set $R \subseteq V \otimes_{k} V$.

\begin{Def}

    The quadratic dual algebra of $A$, denoted $A^!$, is defined as 
    \be
        A^! = T(V^*[-1])/\langle R^{\perp}\rangle
    \ee
    where $R^{\perp} = \{\alpha \in (V^*)^{\otimes 2}~\vert~\alpha(r) = 0 ~\text{for all}~ r \in R\}$ denotes the perpendicular subspace of $R$. 
    
\end{Def}

Let $A$ be a quadratic algebra, and $A^!$ its quadratic dual. We now consider the complex built from $A \otimes (A^!)^*$. Note that since $A^!$ could be infinite-dimensional, we must specify what we mean by $(A^!)^*$. Since $A^!$ is the quotient of the tensor algebra $T(V^*[-1])$, for each $d\in \Z$, the subspace of $(A^!)_d$ with weight $d$ is finite dimensional. We define $(A^!)^*$ to be the $\C^\times$-equivariant dual. Under this definition, $A \otimes (A^!)^*$ is a $\C^\times$-equivariant vector space of the form
\be
    A\otimes(A^!)^* = \bigoplus_{i,j \geq 0}  A_i\otimes(A^!_j)^*
\ee
Consider the following element
\be
\alpha=\sum (x_a)_R\otimes (y^a)_L, \qquad \{x_a\}\subseteq V \text{ basis with dual basis } \{y^a\}\subseteq V^*[-1].
\ee
 Here $(x_a)_R$ denotes the right multiplication by $x_a$ and $(y^a)_L$ denotes the dual of the left multiplication on $A^!$. The following was proven in \cite{BGSkoszul}. 

 \begin{Lem}
 
     $\alpha$ is an MC element, namely it satisfies $d\alpha=\alpha^2=0$, where $d$ is the differential on $ A\otimes(A^!)^*$.
     
     \end{Lem}

We call the complex $(A\otimes (A^!)^*, d+\alpha)$ the Koszul complex of the quadratic algebra $A$. The following is the important result we need to use. 

\begin{Thm}[\cite{BGSkoszul}]

    Any Koszul ring is quadratic. A quadratic ring is Koszul if and only if the Koszul complex is a resolution of $\C$. In this case
    \be
A^!=\mathrm{REnd}_{A\Mod}^{\C^\times} (\C),
    \ee
    where $\mathrm{REnd}^{\C^\times}$ means derived endomorphisms that are weighted under $\C^\times$ (or $\C^\times$-enriched endomorphisms). 
    
\end{Thm}

\begin{Rem}

    The reason we need $\C^\times$-enriched homs is that otherwise the endomorphism space could be a completion of $A^!$ by some topology, related to the fact that the linear dual of $A^!$ is bigger than its $\C^\times$-equivariant dual. 
    
\end{Rem}

\subsection{Example: universal enveloping algebra}\label{subsec:universal}

Let us consider a relevant example that encapsulates the previous section. Consider an $n$-dimensional graded Lie algebra $\mf d$, and its universal enveloping algebra $U(\mf d)$. 
\begin{Prop}
    The quadratic dual of $U_\hbar (\fd)$ is the algebra generated by $\hbar^!, \epsilon^i=x_i^!$ with relations
    \be
        \langle (\hbar^!)^2, [\hbar^!,\epsilon^i] + (-1)^{(|\epsilon^j|+1)|\epsilon^k|}1/2 f_{jk}^i\epsilon^j\epsilon^k, [\epsilon^i,\epsilon^j]  \rangle 
    \ee
\end{Prop}

\begin{proof}
    We have seen that $U_\hbar (\fd)$ is a quadratic algebra with relations $R$ spanned by
    \be\label{eq:Urelations}
        \langle  x \otimes  y - (-1)^{|x||y|} y \otimes  x - \hbar \otimes [x,y], ~ x \otimes \hbar - \hbar \otimes x\rangle 
    \ee
    We need to find $R^\perp\subseteq \fg^!\oplus \C\hbar^!$. Note that $\mathrm{dim}(R)=\textstyle \binom{n}{2} + n$ so $\mathrm{dim}(R^\perp)=(n+1)(n+2)/2$. We claim that the following span $R^\perp$
    \be
        \langle (\hbar^!)^2, [\hbar^!,\epsilon^i] +(-1)^{(|\epsilon^j|+1)|\epsilon^k|} 1/2f_{jk}^i\epsilon^j\epsilon^k, [\epsilon^i,\epsilon^j]  \rangle.
    \ee
    Here, all the commutators are graded commutators. It is clear that they are linearly independent and have the correct dimension; we just need to show that they pair trivially with $R$.

    It is clear for $(\hbar^!)^2$. Consider the relation $[\hbar^!,\epsilon^i] =- (-1)^{(|\epsilon^j|+1)|\epsilon^k|}1/2f_{jk}^i\epsilon^j\epsilon^k$. The RHS pairs trivially with $x \otimes \hbar-\hbar\otimes x$, where the pairing of the LHS with this element reads
    \be
-\langle \hbar^!\otimes \epsilon^i, \hbar \otimes x\rangle - (-1)^{|\epsilon^i|}\langle \epsilon^i\otimes \hbar^!, x\otimes \hbar\rangle=-\langle \hbar^!, \hbar \rangle \langle \epsilon^i, x\rangle- (-1)^{|\epsilon^i|}(-1)^{|x|}\langle \epsilon^i, x\rangle\langle \hbar^!, \hbar\rangle, 
    \ee
where the $(-1)^{|x|}$ comes from exchanging $1\otimes \hbar^!$ with $x\otimes 1$. This is non-zero only if $|x|=|\epsilon^i|+1$ in which case the above vanish. 

Consider now the pairing between $[\hbar^!,\epsilon^i] =- 1/2(-1)^{(|\epsilon^j|+1)|\epsilon^k|}f_{jk}^i\epsilon^j\epsilon^k$ and the first relation of equation \eqref{eq:Urelations}. The pairing of LHS receives a contribution of the form
\be
-\langle \hbar^!\otimes \epsilon^i, \hbar \otimes [x,y]\rangle=- \langle \epsilon^i, [x,y]\rangle. 
\ee
The pairing of the RHS is of the form
\be
\begin{aligned}
&-\frac{1}{2}(-1)^{(|\epsilon^j|+1)|\epsilon^k|}f^i_{jk}\left((-1)^{|\epsilon^k||x|} \langle \epsilon^j, x\rangle \langle \epsilon^k, y\rangle-(-1)^{|x||y|} (-1)^{|\epsilon^k||y|}\langle \epsilon^j, y\rangle \langle \epsilon^k, x\rangle\right)\\ &=-\frac{1}{2}\left( \langle \epsilon^i, [x,y]\rangle-(-1)^{|x||y|}\langle \epsilon^i, [y,x]\rangle\right)=-\langle \epsilon^i, [x,y]\rangle
\end{aligned}
\ee
Here we used that the pairing is non-zero only if $|\epsilon^k|=|y|+1$ in the first term and $|\epsilon^k|=|x|+1$ in the second. The two sides agree, so their difference pair trivially with the first term.

Lastly, consider the relation $[\epsilon^i,\epsilon^j] = 0$. This relations again pairs trivially with the second relation. Pairing with the first relation, it only receives contribution from pairing with $x\otimes y-(-1)^{|x||y|}y\otimes x$. This must be trivial since one of these relations is graded symmetric and another is graded antisymmetric. This completes the proof. 
    
\end{proof}

\begin{Rem}

    Equipping $U_\hbar (\fd)^!$ with the internal differential $d = \text{ad}_{\hbar^!}$ yields $d(\epsilon^i) =-(-1)^{(|\epsilon^j|+1)|\epsilon^k|} 1/2f_{jk}^i\epsilon^j\epsilon^k$ and $d^2 = 0$. This is exactly the Chevalley-Eilenberg differential $d_{\text{CE}}$ on $\text{CE}^*(\mf d) = \Sym(\mf d^*[-1])$. Following the method of periodic localization, in which we set $\hbar=1$ on one side and view the other side as a DG algebra, we recover the Chevalley-Eilenberg complex of $U(\fd)$. 
    
\end{Rem}

\section{Quantization of symplectic reductions}\label{sec:higgskoszul}

In this section, we apply quadratic duality to the set-up of Section \ref{sec:higgs}. Recall from this section, we explained that one would like to study $U(\fd)/(C+J)$ and compute the derived endomorphism of the trivial module. We turn this into a quadratic algebra by considering the algebra
\be\label{eq:quadraticA}
A:=U_\hbar (\fd)[d, \psi]/ \{d, \psi\}=C+\hbar J.
\ee
Here $d$ and $\psi$ are both in degree $1$ and weight $1$. This is a quadratic algebra by definition. In this section, we will compute its quadratic dual, and use its Koszul complex to construct a resolution of $\C$.

\subsection{Koszul dual and periodic localization}

Let $A$ be the algebra defined in equation \eqref{eq:quadraticA}. Let $C=\frac{1}{2}\sum c^{ij}x_i\otimes x_j$ where $c^{ij}=(-1)^{|x_i||x_j|}c^{ji}$. 

\begin{Prop}
    The quadratic dual $A^!$ of $A$ is the algebra generated by $\hbar^!, \epsilon^a=x_a^!, d^!, \psi^!$ with relations
    \be\label{eq:comAdual}
    \begin{split}
        \langle (\hbar^!)^2, [\hbar^!, d^!], [\hbar^!,\psi^!], [d^!, \psi^!], [\psi^!,\epsilon^i], [d^!,\epsilon^i], \hspace{15em} \\ [\epsilon^i, \epsilon^j] +(-1)^{|\epsilon^j|(|\epsilon^i|-1)} c^{ij} d^!\otimes \psi^!, [\hbar^!,\epsilon^i] +(-1)^{(|\epsilon^j|+1)|\epsilon^k|} \textstyle 1/2f^{i}_{jk}\epsilon^j\epsilon^k +\langle \epsilon^i, J\rangle d^!\otimes \psi^!\rangle 
    \end{split}
    \ee
\end{Prop}

\begin{proof}
    We wish to find $R^{\perp} \subseteq \fd^! \oplus \hbar^! \oplus d^! \oplus \psi^!$. Since $\mathrm{dim}(R) = 1/2(n^2+5n+10)$, $\mathrm{dim}(R^{\perp}) = 1/2(n^2+7n + 8)$. We claim that the relations in equation (\ref{eq:comAdual}) span $R^\perp$. It is clear that they are linearly independent and have the correct dimension. Consider first the relations $(\hbar^!)^2, [\hbar^!, d^!], [\hbar^!,\psi^!], [d^!, \psi^!], [\psi^!,\epsilon^i]$, and $[d^!,\epsilon^i]$. It's easy to see that these pair trivially with every relation, and is in fact in $R^{\perp}$. Next, consider the relation $[\epsilon^i, \epsilon^j] +  (-1)^{|\epsilon^j|(|\epsilon^i|-1)}c^{ij}d^!\otimes \psi^!$. This pairs trivially with everything except $[d,\psi] = C+\hbar\otimes J$, and gives
\be
    \begin{split}
    \langle [\epsilon^i, \epsilon^j] +& (-1)^{|\epsilon^j|(|\epsilon^i|-1)}c^{ij}d^!\otimes \psi^!, [d, \psi] - \textstyle \sum c^{ij}x_i \otimes x_j-\hbar\otimes J\rangle\\ &=  (-1)^{|\epsilon^j|(|\epsilon^i|-1)}c^{ij}-\frac{1}{2}\left(\sum_{kl} c^{kl}\langle \epsilon^i\epsilon^j ,x_kx_l\rangle - (-1)^{|\epsilon^i||\epsilon^j|}\sum c^{kl}\langle \epsilon^j\epsilon^i ,x_kx_l\rangle\right) \\
    &=  (-1)^{|\epsilon^j|(|\epsilon^i|-1)}c^{ij}-\frac{1}{2}\left((-1)^{|\epsilon^j|(|\epsilon^i|-1)}c^{ij}- (-1)^{|\epsilon^i|}c^{ji} \right)\\
    &= (-1)^{|\epsilon^j|(|\epsilon^i|-1)} c^{ij}-\frac{1}{2}\left((-1)^{|\epsilon^j|(|\epsilon^i|-1)}c^{ij}- (-1)^{|\epsilon^i|}(-1)^{(|\epsilon^i|-1)(|\epsilon^j|-1)}c^{ij} \right)\\
    &=(-1)^{|\epsilon^j|(|\epsilon^i|-1)} c^{ij}-\left( (-1)^{|\epsilon^j|(|\epsilon^i|-1)}c^{ij}\right)=0
    \end{split}
\ee
Next, let us again consider a relation from the previous problem, $[\hbar^!, \epsilon^i] = -(-1)^{(|\epsilon^j|+1)|\epsilon^k|}1/2f^i_{jk}\epsilon^j\epsilon^k$. This pairs trivially with everything, but unfortunately does not pair trivially with $[d,\psi] = C+\hbar J$.  
\be
\langle \textstyle [\hbar^!, \epsilon^i] +(-1)^{(|\epsilon^j|+1)|\epsilon^k|} \frac{1}{2}f^i_{jk}\epsilon^j\epsilon^k, [d,\psi] -C-\hbar\otimes J \rangle = -(-1)^{(|\epsilon^j|+1)|\epsilon^k|}\langle \frac{1}{2}f^i_{jk}\epsilon^j\epsilon^k, C\rangle-\langle \hbar^!\otimes\epsilon^i, \hbar\otimes J\rangle=-\langle \epsilon^i, J\rangle. 
\ee
Here we used that
\be
f^i_{jk}c^{jk}=-f^{i}_{kj}(-1)^{|x_j||x_k|} (-1)^{|x_k||x_j|}c^{kj}=-f^i_{jk}c^{jk}.
\ee
To fix this, we simply add an extra term $\langle \epsilon^i, J\rangle d^!\otimes \psi^!$, and we find the last relation
\be
[\hbar^!, \epsilon^i] +(-1)^{(|\epsilon^j|+1)|\epsilon^k|} \frac{1}{2}f^i_{jk}\epsilon^j\epsilon^k+\langle \epsilon^i, J\rangle d^!\otimes \psi^!.
\ee
This completes the proof. 

\end{proof}

The algebra $A^!$ has two central elements $d^!$ and $\psi^!$. If we denote by $\varphi=\psi^!$ and view $d^!$ as a variable, then $A^!$ is a deformation of $U_\hbar (\fg)^!$, where the commutation relation of $\epsilon^i$ is deformed by $C$, and the action of $\hbar^!$ is deformed by $\hbar J$. We now set $d^!=1$ and treat $\hbar^!$ as a differential on the remaining generators, and denote the resulting DG algebra by $\CC\CE(\fd, J)$. From the commutation relations of equation \eqref{eq:comAdual}, we immediately obtain the following corollary. 

\begin{Cor}
    The algebra $\CC\CE(\fd, J)$ is a quantization of the Poisson algebra $\mathrm{CE}^*(\fd)$ whose Poisson structure is induced by $C$, or equivalently, by $\kappa$. 
\end{Cor}

We would like to show that $\CC\CE (\fd, J)$ is equal to $\End_{U(\fd)/(C+J)}(\C)$. To do so, we would like to utilize the Koszul complex and show that it is a resolution of $\C$. However, we do not want to directly work with the algebras $A$ and $A^!$. Instead, we perform periodic localization as follows.

Consider $(A^!)^\vee:=\Hom^{\C^\times}_{\C[d^!]}(A^!, \C[d^!])$. That is, we take the graded dual of $A^!$ as a module of $\C[d^!]$. Since $d^!$ is central, the Koszul differential defines a differential on $A\otimes (A^!)^\vee$. Let us consider the subcomplex
\be
K:=dA\otimes \hbar^!(A^!)^\vee\subseteq A\otimes (A^!)^\vee.
\ee
This is obviously a sub-complex. As a vector space, we can identify this complex with
\be
K=U_\hbar (\fd)[\psi]\otimes \hbar^! (A^!)^\vee.
\ee
There is an obvious map $K\to \C[\hbar]\otimes \C[d^!]$, given by sending all the other generators to zero. We prove the following result.

\begin{Prop}
    The map $K\to  \C[\hbar]\otimes \C[d^!]$ is a quasi-isomorphism. 
\end{Prop}

\begin{proof}
 We use a spectral sequence argument. Let us consider the following filtration on $K$
    \be
F_iK=\sum_{j+k=i}\hbar^j\otimes (d^!)^k K.
    \ee
    This is clearly a filtration, since the differential can only raise degrees in $\hbar$ or $d^!$. Taking the associated graded, anything involving $\hbar$ or $d^!$ is set to zero. This includes commutation relations in $U_\hbar (\fd)$ and commutation relations in $A^!$, as well as the $\hbar\otimes \hbar^!$ and $d\otimes d^!$ part of the differential. We find that the associated-graded of $F_iK$ is identified with
    \be
\mathrm{Sym}(\fd\oplus \C\psi)\otimes \bigwedge {}^*(\fd^!\oplus\psi^!)^*\otimes \otimes \C[\hbar]\otimes \C[d^!],
    \ee
    where the differential is just the Koszul differential of $\mathrm{Sym}(\fd\oplus \C\psi)$. Therefore, the cohomology of $\mathrm{Gr}(K)$ is simply $\C[\hbar]\otimes \C[d^!]$. This is the $E_1$-page of the spectral sequence. However, since $\hbar$ and $d^!$ are in degree $0$, no differential appears in this page and the spectral sequence terminates. Consequently, $K\cong \C[\hbar]\otimes \C[d^!]$. 
    
\end{proof}

From this result, we find that $K\vert_{\hbar=d^!=1}\cong \C$. Carefully examine the complex $K\vert_{\hbar=d^!=1}$, we find that it can be identified with
\be
K\vert_{\hbar=d^!=1}=U (\fd)[\psi]\otimes \CC\CE(\fd, J)^*, \qquad d=d\otimes 1+1\otimes d+\sum x_i\otimes \epsilon^i+\psi\otimes \varphi. 
\ee
This is the resolution of $\C$ that we claimed in Section \ref{subsec:mathintro}. As a consequence, we find
\be
\CC\CE (\fd, J)\cong \Hom_{U(\fd)[\psi]\Mod}^{\C^\times}(K, \C)\cong \Hom_{U(\fd)[\psi]\Mod}^{\C^\times}(K, K).
\ee

\subsection{Main example: symplectic reduction}

Now we specialize to the construction from Section \ref{symred}, with the Lie algebra $\fh:=\fg\ltimes V[-1]$ and $\fd=T^*\fh=\tgv$. Recall the commutation relations from equations \eqref{eq:commuth} and \eqref{eq:commutd}. The pairing is given by
\be
\kappa (x_i, t^j)=\delta_i^j, \qquad \kappa (\theta_m, \overline{\theta}^n)=\delta_m^n,
\ee
and the quadratic Casimir element $C=\frac{1}{2}(\sum x_i\otimes t^i+t^i\otimes x_i+\theta_m\otimes \overline{\theta}^m-\overline{\theta}^m\otimes \theta_m)$. We denote $J=\sum J_it^i$. 

\begin{Cor}
    The algebra $\CC\CE(\tgv, J)$ is isomorphic to $\C_\varphi^J[\widehat \CM_H]$.
\end{Cor}

\begin{proof}
  We find that the commutation relation in $\CC\CE(\tgv, J)$ reads
    \be
[\beta_m, \gamma^n]= -\delta^n_m \varphi,\qquad \{b_i, c^j\}=-\delta_i^j \varphi,
    \ee
    and the differential is unchanged on all generators except $b_i$, which is given by
    \be
db_i= -\frac{1}{2}\sum f_{ij}^k (c^j b_k-b_k c^j)-(x_iv, v^*)- J_i \varphi. 
\ee
These are identical to what we find in $\C_\varphi^J[\widehat\CM_H]$ (\textit{c.f.} Section \ref{subsec:higgsquan}).

\end{proof}

\bibliographystyle{amsalpha} 

\bibliography{PI}

\end{document}